\documentclass[12pt]{llncs}

\usepackage{a4wide,amssymb,amsmath,graphicx,epsfig,color,colortbl}

\newcommand{\degree}{\text{deg}}

\newcommand{\LR}{\text{LMR}}
\newcommand{\LB}{\text{LMR}^{\text{HBR}}}
\newcommand{\LO}{\text{LMR}^{\text{SP}}}
\newcommand{\GR}{\text{GEO}}
\newcommand{\GB}{\text{GEO}^{\text{HBR}}}
\newcommand{\GO}{\text{GEO}^{\text{SP}}}
\newcommand{\BR}{\text{HBR}}
\newcommand{\OR}{\text{SP}}

\definecolor{c000}{rgb}{0.50,1.00,0.70}
\definecolor{c001}{rgb}{0.51,1.00,0.70}
\definecolor{c002}{rgb}{0.52,1.00,0.70}
\definecolor{c003}{rgb}{0.53,1.00,0.70}
\definecolor{c004}{rgb}{0.54,1.00,0.70}
\definecolor{c005}{rgb}{0.55,1.00,0.70}
\definecolor{c006}{rgb}{0.56,1.00,0.70}
\definecolor{c007}{rgb}{0.57,1.00,0.70}
\definecolor{c008}{rgb}{0.58,1.00,0.70}
\definecolor{c009}{rgb}{0.59,1.00,0.70}
\definecolor{c010}{rgb}{0.60,1.00,0.70}
\definecolor{c011}{rgb}{0.61,1.00,0.70}
\definecolor{c012}{rgb}{0.62,1.00,0.70}
\definecolor{c013}{rgb}{0.63,1.00,0.70}
\definecolor{c014}{rgb}{0.64,1.00,0.70}
\definecolor{c015}{rgb}{0.65,1.00,0.70}
\definecolor{c016}{rgb}{0.66,1.00,0.70}
\definecolor{c017}{rgb}{0.67,1.00,0.70}
\definecolor{c018}{rgb}{0.68,1.00,0.70}
\definecolor{c019}{rgb}{0.69,1.00,0.70}
\definecolor{c020}{rgb}{0.70,1.00,0.70}
\definecolor{c021}{rgb}{0.71,1.00,0.70}
\definecolor{c022}{rgb}{0.72,1.00,0.70}
\definecolor{c023}{rgb}{0.73,1.00,0.70}
\definecolor{c024}{rgb}{0.74,1.00,0.70}
\definecolor{c025}{rgb}{0.75,1.00,0.70}
\definecolor{c026}{rgb}{0.76,1.00,0.70}
\definecolor{c027}{rgb}{0.77,1.00,0.70}
\definecolor{c028}{rgb}{0.78,1.00,0.70}
\definecolor{c029}{rgb}{0.79,1.00,0.70}
\definecolor{c030}{rgb}{0.80,1.00,0.70}
\definecolor{c031}{rgb}{0.81,1.00,0.70}
\definecolor{c032}{rgb}{0.82,1.00,0.70}
\definecolor{c033}{rgb}{0.83,1.00,0.70}
\definecolor{c034}{rgb}{0.84,1.00,0.70}
\definecolor{c035}{rgb}{0.85,1.00,0.70}
\definecolor{c036}{rgb}{0.86,1.00,0.70}
\definecolor{c037}{rgb}{0.87,1.00,0.70}
\definecolor{c038}{rgb}{0.88,1.00,0.70}
\definecolor{c039}{rgb}{0.89,1.00,0.70}
\definecolor{c040}{rgb}{0.90,1.00,0.70}
\definecolor{c041}{rgb}{0.91,1.00,0.70}
\definecolor{c042}{rgb}{0.92,1.00,0.70}
\definecolor{c043}{rgb}{0.93,1.00,0.70}
\definecolor{c044}{rgb}{0.94,1.00,0.70}
\definecolor{c045}{rgb}{0.95,1.00,0.70}
\definecolor{c046}{rgb}{0.96,1.00,0.70}
\definecolor{c047}{rgb}{0.97,1.00,0.70}
\definecolor{c048}{rgb}{0.98,1.00,0.70}
\definecolor{c049}{rgb}{0.99,1.00,0.70}
\definecolor{c050}{rgb}{1.00,1.00,0.70}
\definecolor{c051}{rgb}{1.00,0.99,0.70}
\definecolor{c052}{rgb}{1.00,0.98,0.70}
\definecolor{c053}{rgb}{1.00,0.97,0.70}
\definecolor{c054}{rgb}{1.00,0.96,0.70}
\definecolor{c055}{rgb}{1.00,0.95,0.70}
\definecolor{c056}{rgb}{1.00,0.94,0.70}
\definecolor{c057}{rgb}{1.00,0.93,0.70}
\definecolor{c058}{rgb}{1.00,0.92,0.70}
\definecolor{c059}{rgb}{1.00,0.91,0.70}
\definecolor{c060}{rgb}{1.00,0.90,0.70}
\definecolor{c061}{rgb}{1.00,0.89,0.70}
\definecolor{c062}{rgb}{1.00,0.88,0.70}
\definecolor{c063}{rgb}{1.00,0.87,0.70}
\definecolor{c064}{rgb}{1.00,0.86,0.70}
\definecolor{c065}{rgb}{1.00,0.85,0.70}
\definecolor{c066}{rgb}{1.00,0.84,0.70}
\definecolor{c067}{rgb}{1.00,0.83,0.70}
\definecolor{c068}{rgb}{1.00,0.82,0.70}
\definecolor{c069}{rgb}{1.00,0.81,0.70}
\definecolor{c070}{rgb}{1.00,0.80,0.70}
\definecolor{c071}{rgb}{1.00,0.79,0.70}
\definecolor{c072}{rgb}{1.00,0.78,0.70}
\definecolor{c073}{rgb}{1.00,0.77,0.70}
\definecolor{c074}{rgb}{1.00,0.76,0.70}
\definecolor{c075}{rgb}{1.00,0.75,0.70}
\definecolor{c076}{rgb}{1.00,0.74,0.70}
\definecolor{c077}{rgb}{1.00,0.73,0.70}
\definecolor{c078}{rgb}{1.00,0.72,0.70}
\definecolor{c079}{rgb}{1.00,0.71,0.70}
\definecolor{c080}{rgb}{1.00,0.70,0.70}
\definecolor{c081}{rgb}{1.00,0.69,0.70}
\definecolor{c082}{rgb}{1.00,0.68,0.70}
\definecolor{c083}{rgb}{1.00,0.67,0.70}
\definecolor{c084}{rgb}{1.00,0.66,0.70}
\definecolor{c085}{rgb}{1.00,0.65,0.70}
\definecolor{c086}{rgb}{1.00,0.64,0.70}
\definecolor{c087}{rgb}{1.00,0.63,0.70}
\definecolor{c088}{rgb}{1.00,0.62,0.70}
\definecolor{c089}{rgb}{1.00,0.61,0.70}
\definecolor{c090}{rgb}{1.00,0.60,0.70}
\definecolor{c091}{rgb}{1.00,0.59,0.70}
\definecolor{c092}{rgb}{1.00,0.58,0.70}
\definecolor{c093}{rgb}{1.00,0.57,0.70}
\definecolor{c094}{rgb}{1.00,0.56,0.70}
\definecolor{c095}{rgb}{1.00,0.55,0.70}
\definecolor{c096}{rgb}{1.00,0.54,0.70}
\definecolor{c097}{rgb}{1.00,0.53,0.70}
\definecolor{c098}{rgb}{1.00,0.52,0.70}
\definecolor{c099}{rgb}{1.00,0.51,0.70}
\definecolor{c100}{rgb}{1.00,0.50,0.70}

\begin{document}

\title{
\vspace*{-2cm}
Hierarchical Bipartition Routing for delivery guarantee in sparse wireless ad hoc sensor networks with obstacles}
\author{Daniel Gau{\ss}mann \and Stefan Hoffmann \and Egon Wanke}
\institute{Institute of Computer Science, Heinrich-Heine-Universit\"at D\"usseldorf, D-40225 D\"usseldorf, Germany}
\date{}
\maketitle

\begin{abstract}
We introduce and evaluate a very simple landmark-based network partition technique called {\em Hierarchical Bipartition Routing} ($\BR$) to support routing with delivery guarantee in wireless ad hoc sensor networks. It is a simple routing protocol that can easily be combined with any other greedy routing algorithm to obtain delivery guarantee. The efficiency of $\BR$ increases if the network is sparse and contains obstacles. The space necessary to store the additional routing information at a node $u$ is on average not larger than the size necessary to store the IDs of the neighbors of $u$. The amount of work to setup the complete data structure is on average proportional to flooding the entire network $\log_2(n)$ times, where $n$ is the total number of sensor nodes. We evaluate the performance of $\BR$ in combination with two simple energy-aware geographic greedy routing algorithms based on physical coordinates and virtual coordinates, respectively. Our simulations show that the difference between using $\BR$ and a weighted shortest path to escape a dead-end is only a few percent in typical cases.
\end{abstract}

\section{Introduction and related work}

A {\em wireless ad hoc sensor network} is a decentralized network not relying on a preexisting infrastructure in that the nodes operate on limited hardware (memory and energy) and can only interchange packets within a radio range. The number of deployed nodes could be very large. Wireless ad hoc sensor networks are receiving a lot of attention in recent years due to their potential applications in various areas such as monitoring, security and data gathering.

{\em Routing} is the process of sending a packet from one or more source nodes to one or more target nodes through intermediate nodes. Each node decides locally to which neighbor the packet is forwarded. The decision is determined by the routing algorithm based on the ID of the destination node, the local topology (and possibly geometry) of the network, extra information stored in each node about the routes (the routing tables) and information contained in the packet itself.

A {\em routing protocol} defines the rules for exchanging the information between nodes.
In {\em geographic routing} protocols the decision to which neighbor the packet is sent is controlled by the position of the nodes and the distances between them. The position information to each node can be obtained either by devices such as GPS or Galileo (geographic coordinates) or by analyzing the network structure (virtual coordinates). Position awareness can often significantly improve the efficiency of routing. In \cite{MRSS08} it is mentioned that protocols using position information for routing like MFR \cite{TK84}, COP \cite{Sto06}, and GFG \cite{FS06} are competitive alternatives to the classical routing protocols for wireless ad hoc networks as for example DSR \cite{JMB01}, AODV \cite{PBD03}, and OLSR \cite{CJLMQV03}). 

Most algorithms based on position awareness first try to deliver the packet using greedy techniques. For example, the simplest greedy routing technique \cite{Fin87} will forward the packet to the neighbor closest to the destination. {\em Most Forward Routing} (MFR) \cite{TK84} and {\em Nearest with Forwarding Progress} (NFP) \cite{HL86} consider the projected distance on the source-destination line. MFR tries to get closer to the destination by sending the packet to the neighbor with maximum projected distance, while NFP suggests to adjust the transmission power by sending the packet to the neighbor with the smallest projected distance.

Greedy routing algorithms can easily be extended by taking into account power consumption. Power aware greedy routing in most cases tries to minimize the ratio of the energy consumed by a transmission to the progress made. The progress is the distance reduction towards the destination. This {\em cost over progress} (COP) power-aware framework is first introduced in \cite{KNS06}. If the cost is equal for all connections, we obtain the simple greedy algorithm as already discussed above \cite{Fin87}. If the cost of a connection is proportional to the distance between the nodes, the resulting routing is similar to compass routing \cite{KSU99}.

In general, greedy routing algorithms do not guarantee packet delivery. A packet can be trapped in a local minimum where the algorithm will fail to find a next neighbor. The probability of reaching a so-called {\em dead-end} increases if the network is less dense or if the network contains obstacles where no nodes can be placed and/or connections are truncated by obstacles.

There are several attempts to obtain delivery guarantee for greedy routing algorithms. The authors of \cite{BMSU01} propose {\em face routing}, which guarantees delivery in two-dimensional unit disk graphs (UDG). Face routing is applied to a planar sub-network obtained by considering the Gabriel Graph \cite{BMSU01,LGSM04}, the Relative Neighborhood Graph \cite{Tou80}, or the Morelia Graph \cite{BCGKOSU04}. In \cite{SD04} a greedy-face-greedy (GFG) approach is considered, where greedy routing is based on COP as in \cite{Sto06} and face routing is similar to the one in \cite{BMSU01}. Energy-aware routing is also proposed in LEARN \cite{WSWLD06}, SPFSP \cite{SR06}, End-to-End (EtE) \cite{EMS08}, and EEGR \cite{ZS07}.

Landmark-based routing algorithms like VCap \cite{CCDU05}, JUMPS \cite{BPDCFS06}, GLIDER \cite{FGGSZ05}, VCost \cite{EMS07}, and BVR \cite{FRZ05} use virtual coordinates computed from the distances to specific nodes called {\em landmarks}, {\em anchors}, or {\em beacons}. In the first phase, a global and distributed election mechanism elects a set of nodes acting as landmarks. Then the landmarks flood the entire network or only parts of the network such that every node can compute its virtual coordinate depending on the distances to the landmarks. The virtual coordinates can then be used to route a message greedily through the network. Packet delivery is also not guaranteed if different nodes have the same virtual coordinates.

There are also several attempts to obtain delivery guarantee for landmark-based greedy algorithms. Most of them are based on a tree coordinate system like LTP \cite{CMT07} and ABVCap \cite{LA08}. An energy efficient approach is introduced in HECTOR \cite{MRSS08}. This protocol mixes the use of the tree-based coordinate system of \cite{CMT07} and the landmark-based coordinate system of \cite{CCDU05} and \cite{EMS07}.

An alternative way for delivery guarantee can be obtained by hierarchical addressing, see for example \cite{Tsu88}. Tsuchiya solves this problem by allowing nodes to self– configure their addresses. The protocol uses a hierarchical set of landmark nodes that periodically send scoped route discovery messages. A node's address is the concatenation of its closest landmark at each level in the hierarchy. The overhead of route setup can be reduced to $O(n \log n)$ and nodes only hold state for their immediate neighbors and their next hop to each landmark. However, this requires a protocol that creates and maintains this hierarchy of landmarks and appropriately tunes the landmark scopes. Recent proposals adopting this approach have been fairly complex \cite{KAE00} in contrast to our design goal of configuration simplicity, see also \cite{IS09} for an overview.

A second alternative to obtain delivery guarantee is clustering which is proposed by various researchers as for example in \cite{CWLG97} and \cite{KCVP95}. A closely related approach is the construction of connected dominating sets as routing backbones \cite{W02}.

In this paper, we introduce and evaluate a very simple landmark-based network partition technique called {\em Hierarchical Bipartition Routing} ($\BR$) to support routing with delivery guarantee in wireless ad hoc sensor networks. It is a simple routing protocol that can easily be combined with any other greedy routing algorithm to obtain delivery guarantee. The efficiency of $\BR$ increases if the network is sparse and contains obstacles. These are exactly the networks where greedy algorithms will fail with high probability. The hierarchical bipartition of the network is performed on a landmark-based data structure setup in a pre-processing phase.

The space necessary to store the additional routing information at a node $u$ is approximately $\log_2(n) \cdot \log_2(\deg(u))$ bits on average, where $n$ is the total number of sensor nodes in the network and $\deg(u)$ is the number of neighbors of $u$. This is in general not larger than the size necessary to store the IDs of all neighbors of $u$. The amount of work to setup the complete data structure is on average proportional to flooding the entire network $\log_2(n)$ times. The sizes of the virtual addresses are on average only a few bits larger than the sizes of the IDs.

We evaluate the performance of $\BR$ in combination with two simple energy-aware geographic greedy routing algorithms based on physical coordinates and virtual coordinates, respectively. Our simulations show that the difference between using $\BR$ and a weighted shortest path to escape a dead-end is only a few percent in typical cases.

This paper is organized as follows. In the next section, we define $\BR$. After that in Section 3, we describe two simple greedy routing protocols and how they can escape from a dead-end using $\BR$ or a weighted shortest path. In section 4, we evaluate the performance of $\BR$. Conclusions are given in Section 5.

\section {Hierarchical Bipartition Routing ($\BR$)}

A {\em network} is modeled as an {\em undirected graph} $G=(V,E)$, where $V$ is the set of {\em nodes} ({\em sensor nodes}) and $E \subseteq \{\{u,v\}~|~u,v \in V,~u \not=v\}$ is the set of undirected {\em edges} ({\em connections}) between sensor nodes. A connection $e=\{u,v\}$ may have a positive {\em weight} $\omega(e) \in {\mathbb R}_+$, also denoted by $\omega (u,v)$ or $\omega (v,u)$. In general, the weight represents the amount of energy necessary to reach the neighbor node. A {\em path} $p= u_1,\ldots,u_k$ is a non-empty sequence of nodes $u_i \in V$, $1 \leq i \leq k$, such that $\{u_i,u_{i+1}\} \in E$, $1 \leq i < k$, is a connection between $u_i$ and $u_{i+1}$. The {\em weight} of $p$ is \[\omega(p)= \sum_{i=1}^{k-1} \omega(u_i,u_{i+1}).\] Path $p$ is a shortest path between $u$ and $v$, if there is no path $p'$ between $u$ and $v$ with $\omega(p') < \omega(p)$.

The $\omega$-{\em distance} $d^{\omega}(u,v)$ between two nodes is the weight of a shortest path between $u$ and $v$.

The {\em hop distance} $d^h(u,v)$ between two nodes is the weight of a shortest path between $u$ and $v$ for the case that all connections $\{u,v\}$ have weight $\omega(u,v)=1$.

To analyze the performance of $\BR$, we assume that every node $u \in V$ has a physical geographic position $(u_x,u_y) \in {\mathbb R}^2$ in the plane defined by a two-dimensional real vector. These positions are only used by the geographic greedy routing protocols, and not by the hierarchical bipartition technique. The {\em euclidean distance} between two nodes $u$ and $v$ is \[d^e(u,v) = \sqrt{(u_x-v_x)^2+(u_y-v_y)^2}.\] Note that the $\omega$-distance $d^\omega(u,v)$ and the hop distance $d^h(u,v)$ are defined by the network structure, whereas the euclidean distance is defined by the physical geographic positions of the nodes. 

A network $G=(V,E)$ is {\em connected} if there is a path between every pair of nodes. Network $G'=(V',E')$ is a {\em sub-network} of $G$ if $V \subseteq V$ and $E' \subseteq E$, it is an {\em induced sub-network} of $G$ if $V' \subseteq V$ and $E' = \{\{u,v\}~|~\{u,v\} \in E,~u,v\in V'\}$.

The nodes are assumed to be static. Each node is assumed to have a unique ID which is mainly used to break ties. The unique ID is also necessary to specify the target node of the packet, if we do not want to give every node a unique virtual address.

\subsection {Initialization}

The necessary data structure is built in several phases. In the first phase, an arbitrary node $w$ and two {\em landmark nodes} $x_0$, $x_1$ are selected. Landmark node $x_0$ is one of the nodes that has maximum $\omega$-distance to $w$, and landmark node $x_1$ is one of the nodes that has maximum $\omega$-distance to $x_0$. Every node $u$ with \[d^{\omega}(u,x_0) \leq d^{\omega}(u,x_1) \text{ ~~~ gets virtual address ~~~ } 0,\] and every node $u$ with \[d^{\omega}(u,x_0) > d^{\omega}(u,x_1) \text{ ~~~ gets virtual address ~~~ } 1.\]

Let $G_0$ be the network induced by the nodes with virtual address 0 and let $G_1$ be the network induced by the nodes with virtual address 1.

\begin{lemma}
\label{L1}
If network $G$ is connected, then $G_0$ and $G_1$ are connected.
\end{lemma}

\begin{proof}
Let $v$ be a neighbor of a node $u$ on a shortest path between $u$ and $x_0$, then
\[\omega(u,v) + d^\omega(v,x_0) = d^\omega(u,x_0).\]
If the virtual address of $u$ is 0, then
\[d^\omega(u,x_0) \leq d^\omega(u,x_1).\]
Since
\[d^\omega(u,x_1) \leq \omega(u,v) + d^\omega(w,x_1),\]
we get
\[d^\omega(v,x_0) \leq d^\omega(v,x_1).\]
That is, all nodes $v$ on a shortest path between $u$ and $x_0$ have virtual address $0$ and thus $G_0$ is connected. An analogous argumentation shows that $G_1$ is connected. $\Box$
\end{proof}

In the next phase, we select for every virtual address $\alpha \in \{0, 1\}$ two landmark nodes $x_{\alpha \cdot 0}$, $x_{\alpha \cdot 1}$ from the connected sub-network $G_{\alpha}$. Here $\alpha \cdot 0$ and $\alpha \cdot 1$ is the extension of $\alpha$ by symbol 0 or 1, respectively. Landmark node $x_{\alpha \cdot 0}$ is one of the nodes of network $G_{\alpha}$ that has maximum $\omega$-distance to $x_{\alpha}$ in sub-network $G_{\alpha}$, and landmark node $x_{\alpha \cdot 1}$ is one of the nodes of network $G_{\alpha}$ that has maximum $\omega$-distance to $x_{\alpha \cdot 0}$ in sub-network $G_{\alpha}$. A node $u$ of $G_{\alpha}$ whose $\omega$-distance to $x_{\alpha \cdot 0}$ is less than or equal to the $\omega$-distance between $u$ and $x_{\alpha \cdot 1}$ gets virtual address $\alpha \cdot 0$. It gets virtual address $\alpha \cdot 1$, if the distance between $u$ and $x_{\alpha \cdot 0}$ is greater than the distance between $u$ and $x_{\alpha \cdot 1}$.

The bipartition of every $G_{\alpha}$ into two further sub-networks $G_{\alpha \cdot 0}$ and $G_{\alpha \cdot 1}$ can be continued with the new virtual addresses $\alpha$ until all the created sub-networks consist of only one single node. In this case, the nodes of the network are uniquely identified by the virtual addresses. An inductive application of Lemma \ref{L1} shows that all the sub-networks $G_{\alpha}$ are connected.

\subsection {Distributed address computation}

The hierarchical bipartition can easily be computed by a distributed algorithm. Once the network is deployed, an arbitrarily selected node $w$ starts flooding the network. The message carries a weight initialized to zero. The weight is increased by every forwarding node. If node $u$ sends a message to a neighbor $v_i$ then $u$ increases the weight by $\omega(u,v_i)$. If a node receives more than one message, it will store and forward only the one with the smaller weight. If a node does not receive a new message for a while, the current weight represents the $\omega$-distance to $w$. To be sure that the flooding is finished, the node has to wait for a time longer than the time required to propagate a message through the network. To reduce the overhead during the distance computation, the messages should be sent to the neighbors $v_1,\ldots,v_m$ in ascending order with respect to the costs at the edges, that is, $\omega(u,v_1) \leq \ldots \leq \omega(u,v_m)$.

The election of the landmark nodes for the bipartition of the network can be done by the following simple protocol. Assume we want to determine a unique node $u$ with maximum $\omega$-distance to some other node $w$. Then all nodes with a maximum $\omega$-distance to $w$ in its two-hop neighborhood (in case of parity the nodes with maximum IDs) start sending a message back to $w$. (The route back to $w$ can be stored during the update process of the $\omega$-distance to $w$.) Node $w$ receives all these messages and can select the node $u$ with maximum distance to $w$. It has to wait for a while such that no further messages will arrive. Then it sends a message back to the winner, the node $u$ with maximum distance to $w$.

\subsection {Routing protocol}

The routing protocol is quite simple and straightforward. Assume a packet should be sent from a source node $s$ to a target node $t$. If the virtual address of $s$ starts with symbol 0 and the virtual address of $t$ starts with symbol 1, then $s$ is in sub-network $G_0$ and $t$ is in sub-network $G_1$. In this case, the packet is sent step by step to a neighbor whose distance to landmark node $x_{1}$ is minimum until it reaches a node in $G_1$. Then it is routed within the connected sub-network $G_1$ using the second symbol of the virtual addresses, and so on.

More generally, let $\alpha \cdot d_u \cdot \alpha_u$ be the virtual address of the current node $u$ and $\alpha \cdot d_t \cdot \alpha_t$ be the virtual address of the destination $t$ such that $\alpha, \alpha_u, \alpha_t \in \{0, 1\}^\star$, $d_u, d_t \in \{0, 1\}$, and $d_u \not = d_t$. That is, the symbols left to $d_u$ and left to $d_t$ are equal in both virtual addresses. If $d_u= 0$ and $d_t = 1$, the packet is sent greedily towards landmark node $x_{\alpha \cdot 1}$, if $d_u= 1$ and $d_t = 0$, the packet is sent greedily towards landmark node $x_{\alpha \cdot 0}$. The packet does not leave the connected sub-network $G_{\alpha}$. An inductive argumentation proves that $\BR$ guarantees delivery.

\begin{corollary}
A packet sent with $\BR$ always reaches its destination.
\end{corollary}

\subsection {Address size}

The sizes $|\alpha|$ of the virtual addresses $\alpha$ depend on the number of partitions necessary to obtain single node gra\-phs. In this case, the virtual addresses are unique and can be used for the identification of the nodes. If the weights of all edges are equal, then a worst case for the address length is a complete network, that is, a network where every node is connected to all other nodes. Then the number of bipartitions and thus the address length is $n-1$. To avoid these worst-case situations, we can stop the bipartition process when all nodes of a sub-network $G_{\alpha}$ have hop distance $\leq 1$ to landmark node $x_{\alpha}$. This will considerably reduce the address size. Since now the nodes do not have unique virtual addresses, it is necessary to include the target ID in the packets.

The routing protocol can be extended as follows: Assume the packet reaches a node $u$ that has virtual target address $\alpha$ but is not the target node. If a neighbor $v$ of $u$ is the target node, the packet can be sent to node $v$. Otherwise, the packet can be sent to $x_{\alpha}$ and from $x_{\alpha}$ to the target node. This is always possible, because all nodes with virtual address $\alpha$ are neighbors of $x_{\alpha}$. The decision to stop the bipartition is very easy to implement, because the nodes of $G_{\alpha}$ that are candidates for flooding only have to check their list of neighbors.

If the weights of the connections are not all equal, but depend on the distances of the connections, then it is in general not necessary to abort the bipartition process. In practical cases, the different lengths of the connections yields to a partition of the network into two almost equally sized sub-networks. In this case, the virtual addresses are unique and it will not be necessary to use the original IDs.

\subsection {Storage size}

A node has to store its ID, its virtual address, a routing table, and temporarily during the initialization phase some $\omega$-distances to landmark nodes and some source IDs. For the partition of sub-network $G_{\alpha}$ into $G_{\alpha \cdot 0}$ and $G_{\alpha \cdot 1}$, we only need the $\omega$-distances to $x_{\alpha \cdot 0}$ and $x_{\alpha \cdot 1}$. When the new virtual addresses $\alpha \cdot 0$ and $\alpha \cdot 1$ are assigned, it is no longer necessary to store these $\omega$-distances, and the IDs of $x_{\alpha \cdot 0}$ and $x_{\alpha \cdot 1}$. For routing a packet it is sufficient to know for every position $i$, $1 \leq i \leq |\alpha|$, the neighbor to which the packet has to be sent if the own address $\alpha$ and the target address of the packet are equal at the first $i-1$ positions and differ at position $i$. If a node $u$ with virtual address $\alpha$ has $\degree(u)$ neighbors, then the size of the additional routing information is only $|\alpha| \cdot \log_2 \deg(u)$ bits.

\subsection {Worst case behavior}

From a theoretical point of view, the weight of a path routed by $\BR$ can be arbitrarily larger than a shortest path between the source and target node. Figure \ref{fig01} shows a simple example. The virtual addresses of the black and white nodes start with symbol 0 and 1, respectively. A shortest path between the source node $s$ and the target node $t$ has weight $2 \cdot a$. $\BR$ routing will send the packet from $s$ to $v$ and then via $x_1$ to $t$. The weight of this path is $m \cdot a$, where $m$ can be arbitrary large. However, this is not a typical case for randomly generated networks.  

\begin{figure}[hbt]
\center
\epsfig{figure=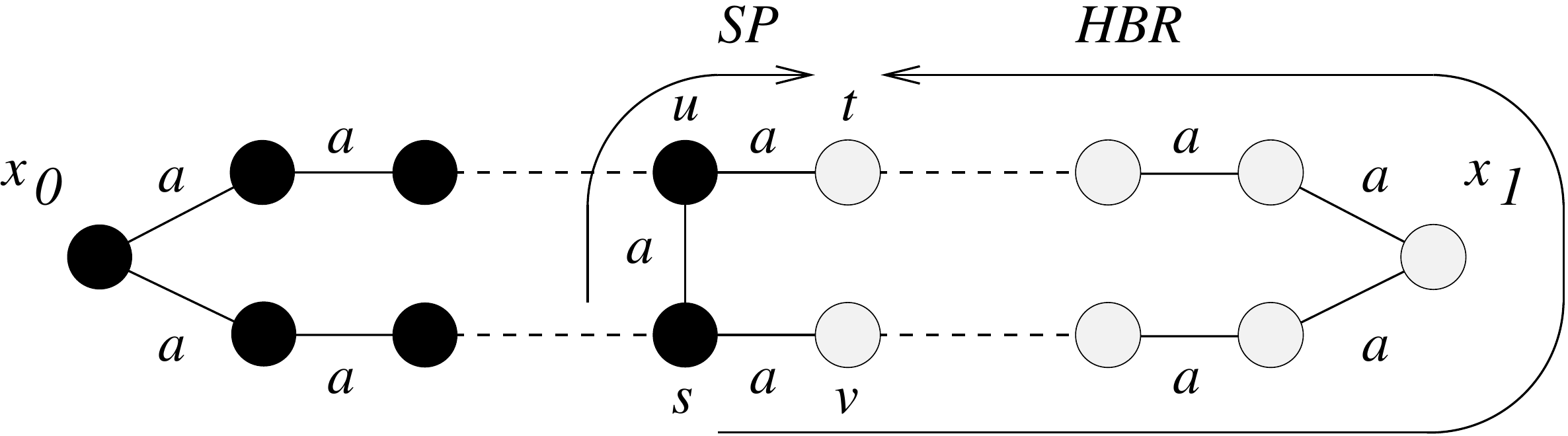,width=200pt} 
\caption{\label{fig01} An unrealistic worst case for the stretch factor of $\BR$}
\end{figure}

In the worst case, the size of the virtual addresses $\alpha$ can reach the number $n$ of nodes. This is for example the case if the network is complete and all edges have the same weight, or if the network is a path and the connections have exponentially increasing weights $1, 2, 4, 8, 16, 32, \ldots $. However, if the number of nodes reachable with increasing distance is for most nodes approximately the same, the size of the virtual addresses is at most $\log_2 n$. This is also confirmed by our experimental evaluations of randomly generated networks.

\begin{figure}[hbt]
\center
\epsfig{figure=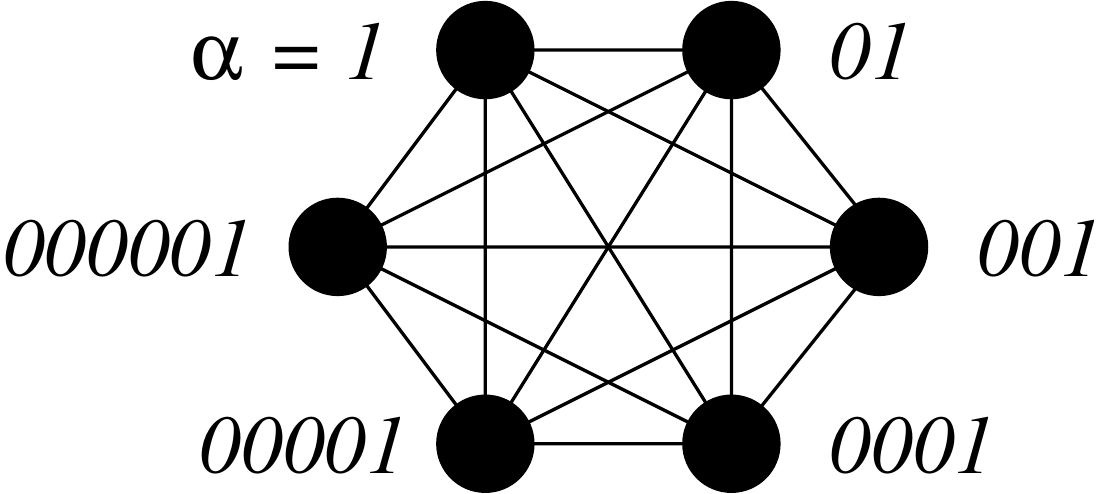,width=100pt}

\bigskip 
\bigskip 
\epsfig{figure=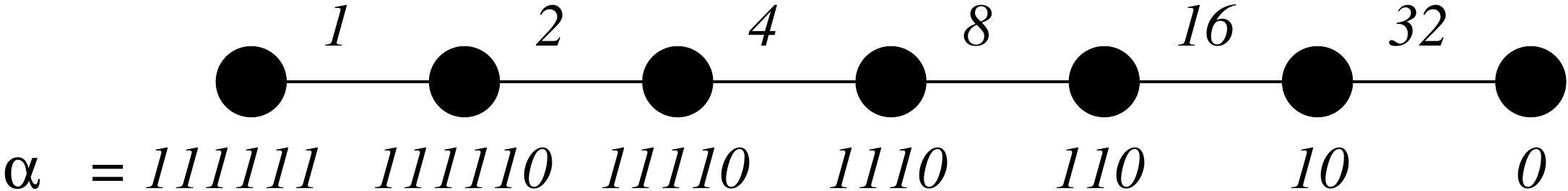,width=190pt} 
\caption{\label{fig02} Two unrealistic worst cases for the size of the virtual addresses}
\end{figure}

\section {Two greedy routing protocols}

We mainly want to use $\BR$ to guarantee delivery for greedy routing protocols. For this reason, we combine $\BR$ with geographic greedy routing based on physical and virtual coordinates.

\subsection {Physical coordinates}

As already mentioned in the introduction, the simplest energy-aware geographic greedy routing protocol sends the packet to a neighbor for which the ratio of cost over progress is minimum. This {\em cost over progress} (COP) power-aware framework is introduced in \cite{KNS06}. If we consider physical coordinates, the packet is sent from node $u$ to a neighbor $v$ of $u$ for which
\[d^e(v,t) < d^e(u,t)\] and
\[\frac{\omega(u,v)}{d^e(u,t)-d^e(v,t)}\]
is minimum. If $\omega(u,v) = a$, $a > 0$, for all connections $\{u,v\}$, we obtain the simple geographic greedy routing that selects a neighbor closest to the destination \cite{Fin87}. For $\omega(u,v) = d^e(u,v)$, the routing is similar to compass routing \cite{KSU99}. If the angle between $(u,v)$ and $(u,t)$ is $\beta$ and the distance to the target node tends to infinity, the ratio of cost over progress tends to $\frac{1}{\cos(\beta)}$. If $\omega(u,v)$ is defined by the commonly used energy function \[a+b \cdot {d^e(u,v)}^c,\] $a,b > 0$, $c \geq 2$, then an optimal routing tries to use equidistant steps towards the target node $t$. The best progress (also called the {\em characteristic distance}) is
\[d^{\star}= \sqrt[c]{\frac{a}{b\cdot(c-1)}},\]
see also \cite{SL01}. The ratio of cost over progress has its minimum at the position with distance $d^{\star}$ from $u$ in the direction to the destination.

\subsection {Virtual coordinates}

Our second greedy routing protocol is based on virtual coordinates which we will define by four landmark nodes denoted by $A$, $B$, $C$, and $D$. These four landmark nodes are selected similarly as in VCap (virtual coordinate assignment protocol) from \cite{CCDU05}. The first landmark node $A$ is one of the nodes with maximum $\omega$-distance to an arbitrary node $w$. The second landmark node $B$ is one of the nodes with maximum $\omega$-distance to $A$. The third landmark node $C$ is one of the nodes for which \[d^{\omega}(C,A)+d^{\omega}(C,B)-2 \cdot |d^{\omega}(C,A)-d^{\omega}(C,B)|\] is maximum. And finally, the fourth landmark node $D$ is one of the nodes for which \[d^{\omega}(D,C)-|d^{\omega}(D,A)-d^{\omega}(D,B)|\] is maximum. Since will consider energy efficient routing, we use the $\omega$-distances instead of hop distances \cite{CCDU05} for the computation of the virtual addresses.

In case of parity, the node with larger ID is chosen. The four landmark nodes $A,B,C,D$ define for every node $u$ a 4-tuple
\[(d^{\omega}(u,A),d^{\omega}(u,B),d^{\omega}(u,C),d^{\omega}(u,D)).\]

Our landmark-based routing protocol sends the packet to a neighbor $v$ for which the ratio cost over progress is minimum. The progress $d'(u,t)-d'(v,t)$ is defined by distance function
\[
d'(u,v) = \sqrt{
\begin{array}{l}
(d^{\omega}(u,A)-d^{\omega}(v,A))^2+
(d^{\omega}(u,B)-d^{\omega}(v,B))^2+ \\
(d^{\omega}(u,C)-d^{\omega}(v,C))^2+
(d^{\omega}(u,D)-d^{\omega}(v,D))^2
\end{array}}.\]

\subsection {Dead-end handling}

Both greedy routing protocols can reach a so-called {\em dead-end}, i.e., a node $u$ that has no neighbor $v$ closer to the destination than $u$. If a dead-end is reached, the packet is either sent along a shortest path or by $\BR$ towards the destination node. In both cases the weight function $\omega$ is applied. The packet is sent hop by hop until a node is reached whose distance to the destination is less than the distance from the last dead-end node to the destination. Then the original greedy routing is continued.

It is obvious that a shortest-path routing is not possible in practice. We use shortest-path routing only to get a comparison with $\BR$ under the assumption that following a shortest path is a good idea to get out of a dead-end.

The two geographic routing variants based on physical coordinates are denoted by $\GO$ and $\GB$, the two variants based on virtual coordinates are denoted by $\LO$ or $\LB$, depending on whether the dead-end problem is cleared with the help of a shortest path or by $\BR$, respectively.

\section {Analysis}

The analysis of $\BR$ is done by randomly generated networks and randomly selected source and target nodes. The test environment and the obtained evaluation results is explained in the next subsections.

\subsection {Experimental environment}

Our networks have a size of $1000 m$ $\times$ $1000 m$. The radio range is fixed at $50 m$, the node density $\delta$ varies between $0.5 \cdot 10^{-3}$ and $9.2 \cdot 10^{-3}$ nodes per $m^2$, which corresponds to an average node degree between $4$ and $72$. If one of the randomly created networks is disconnected, we use the largest connected component, if its size is at least $\frac{2}{3}$ of the size of the complete network.

Networks with holes or obstacles are created with the help of black/white-masks. If the randomly selected position of a node hits a white-entry of the mask, the node is omitted. We do not try to find another position for this node. The masks we use for our evaluations are shown in Figure \ref{fig03}, \ref{fig05}, and \ref{fig07}.

\begin{figure}[hbt]
\center
\hfill
\epsfig{figure=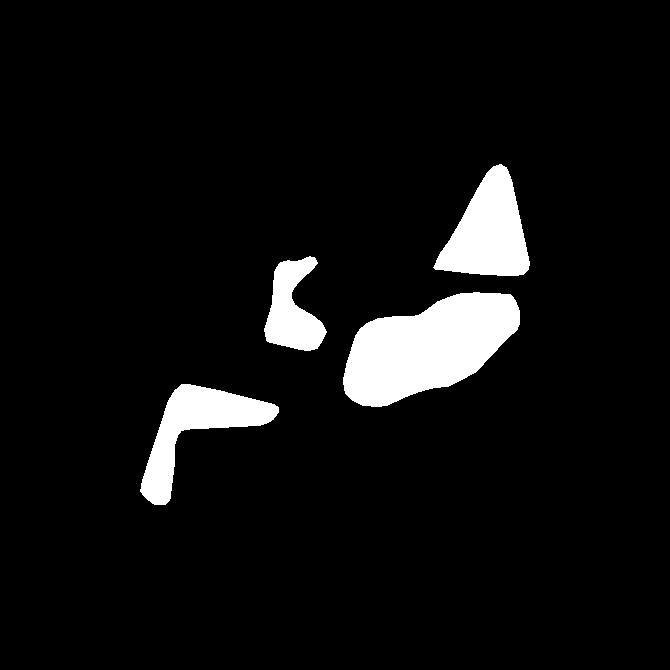,width=200pt} 
\hfill
\epsfig{figure=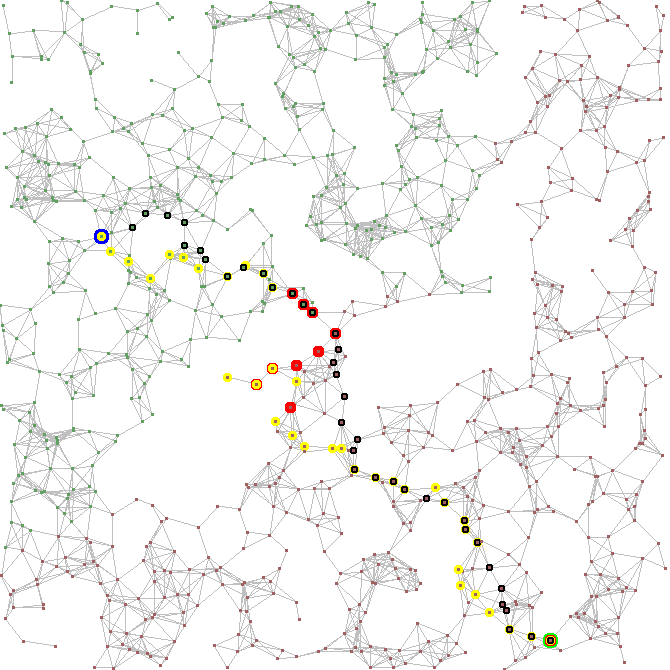,width=200pt} 
\hfill
\null
\caption{\label{fig03} Left: Lakes-mask, lat. $51.19^\circ$, lon. $6.37^\circ$; Right: Routing in a network generated with lakes-mask and density $\delta= 1.0 \cdot 10^{-3}$}
\end{figure}


The lakes-mask of Figure \ref{fig03} (latitude $51.19^\circ$, longitude $6.37^\circ$) represents wet areas where sensors are lost during the dispersion process. The streets-mask of Figure \ref{fig05} (latitude $40.70^\circ$, longitude $-73.93^\circ$) represents an area where the sensor nodes are assumed to be dispersed by vehicles driving along streets. The buildings-mask of Figure \ref{fig07} (latitude $52.50^\circ$, longitude $13.35^\circ$) represents an example of a metropolitan area. Here we additionally remove all connections between sensors that can not see each other, because there is a building in between.

\begin{figure}[hbt]
\center
\hfill
\epsfig{figure=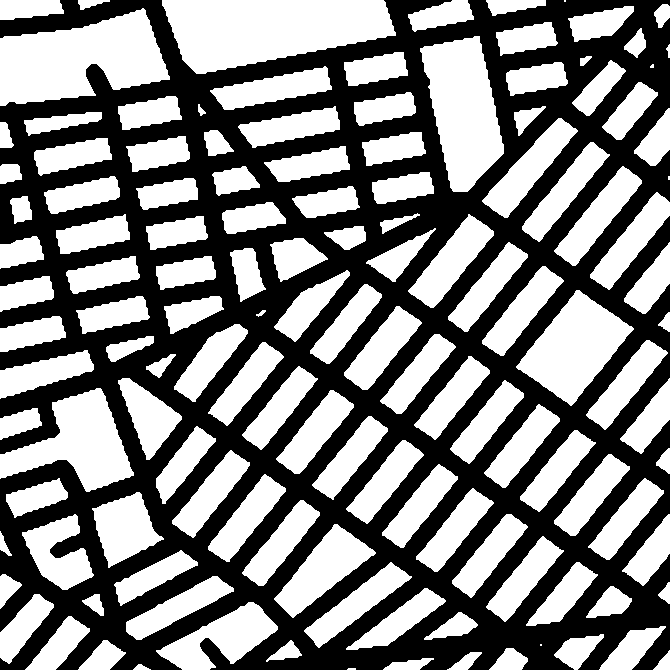,width=200pt}
\hfill
\epsfig{figure=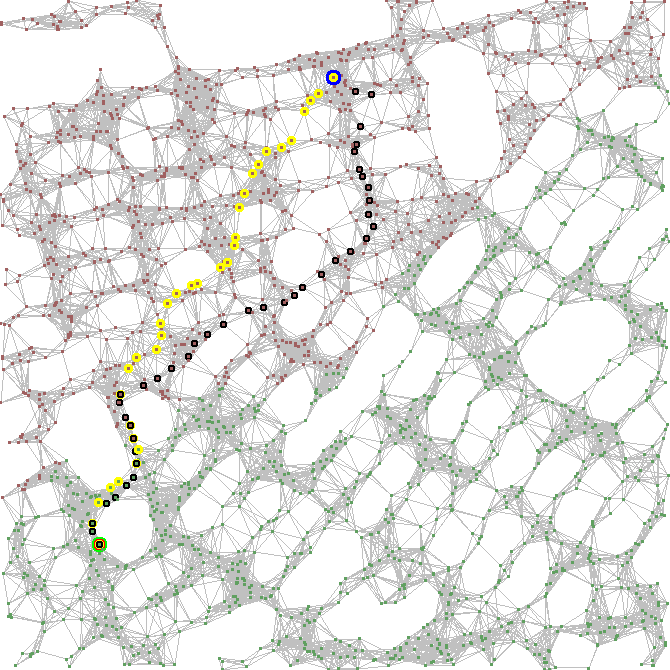,width=200pt} 
\hfill
\null
\caption{\label{fig05} Left: Streets-mask, lat. $40.70^\circ$, lon. $-73.93^\circ$; Right: routing in a network generated with streets-mask and density $\delta= 4.5 \cdot 10^{-3}$}
\end{figure}


We use the energy function
\begin{equation}
\label{equ01}
\omega(u,v) = 400 + d^e(u,v)^2
\end{equation}
for every connection $\{u,v\} \in E$, such that the characteristic distance is $d^{\star} = 20~m$.

Each of the Figures \ref{fig03}, \ref{fig05}, and \ref{fig07} shows to the right two routing paths. The start node is encircled green, the destination node is encircled blue. The nodes traversed by $\BR$ are colored black. The nodes traversed by $\GO$ are colored yellow. The dead-end nodes of $\GO$ are colored red. The small light green (light red) nodes have a virtual address starting with 0 (with 1, respectively).

\begin{figure}[hbt]
\center
\hfill
\epsfig{figure=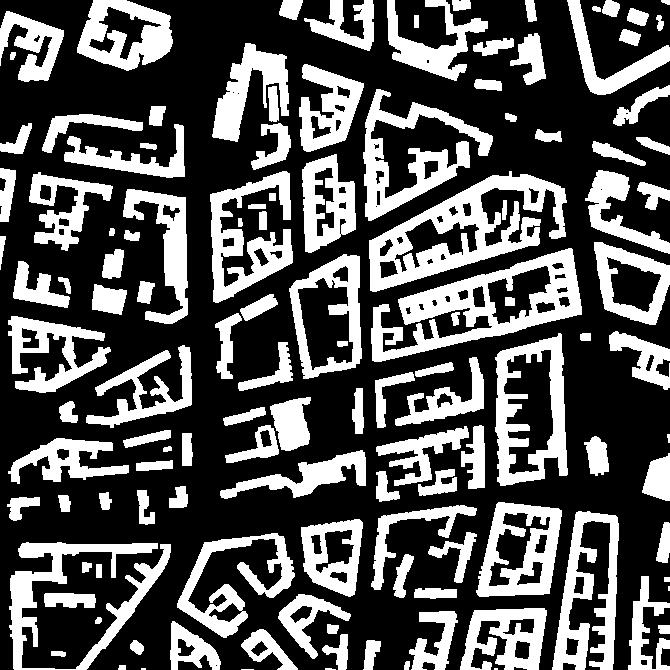,width=200pt}
\hfill
\epsfig{figure=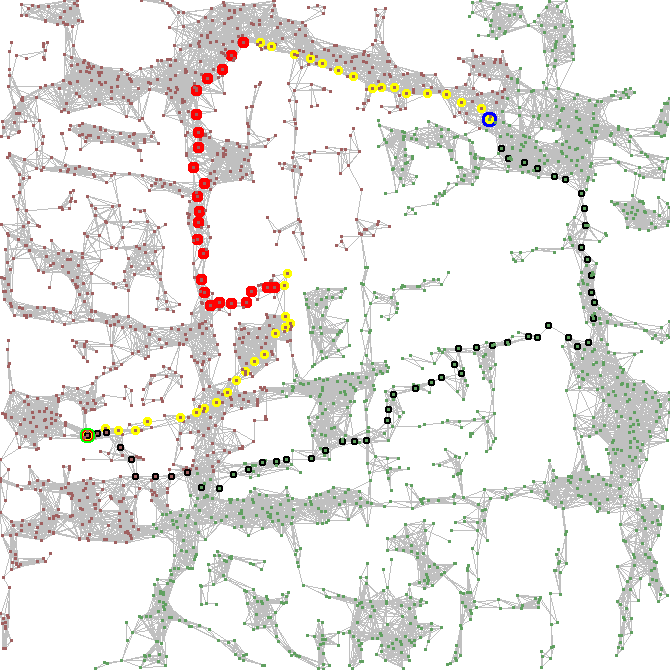,width=200pt} 
\hfill
\null
\caption{\label{fig07} Left: Buildings-mask, lat. $52.50^\circ$, lon. $13.35^\circ$; Right: Routing in a network generated with building-mask, only visible connections, and density $\delta= 4.5 \cdot 10^{-3}$}
\end{figure}


For every node density $\delta$ between $0.5 \cdot 10^{-3}$ and $9.2 \cdot 10^{-3}$ in steps defined by factor $1.2$, we randomly create 1000 networks. For every network, we randomly selected 1000 source nodes and 1000 target nodes. Let $C(\BR)$ and $C(\OR)$ be the sum of the costs to route from the 1000 source nodes to the corresponding 1000 target nodes in all 1000 networks using $\BR$ and $\OR$, respectively. The cost of a route is the sum of the weights of the used connections. The {\em overhead} of $\BR$ is defined by
\[\frac{C(\BR)-C(\OR)}{C(\OR)}.\]
It is defined in the same way for the other routing protocols $\LO$, $\LB$, $\GO$, and $\GB$.

\subsection {Evaluations}

The tables \ref{tab01}, \ref{tab02}, \ref{tab03}, and \ref{tab04} show the average overhead in percent as a function of the node density $\delta$. The overhead entries are colored continuously between green (0\%) and red (50\%). We think that it is more valuable to present the main results in a table than in a graphical illustration, because it is easy to create a graphical view from the values of the table, but not vice versa. 

Table \ref{tab01} considers networks where the sensor nodes are uniformly distributed. If the node density is greater than or equal to $1.0 \cdot 10^{-3}$, the two greedy routing algorithms $\GO$ and $\GB$ on physical coordinates have less overhead than the greedy routing algorithms $\LO$ and $\LB$ on virtual coordinates. For less dense networks ($\delta < 1.0 \cdot 10^{-3}$), $\BR$ has even less overhead than $\GO$, $\GB$, $\LO$, and $\LB$. The difference between a shortest path dead-end handling and a $\BR$ dead-end handling, i.e., the difference between $\GO$ and $\GB$ and between $\LO$ and $\LB$, is only a few percent. This holds for greedy routing with physical coordinates as well as for greedy routing with virtual coordinates.

\begin{table}[hbt]
\center
\small
\hfill
\begin{tabular}{|c|c|c|c|c|c|}
\hline
$\delta$ & $\BR$ & $\LO$ & $\LB$ & $\GO$ & $\GB$ \\
\hline
0.5 & \cellcolor{c005}  2.94 & \cellcolor{c040} 20.06 & \cellcolor{c042} 21.20 & \cellcolor{c052} 26.01 & \cellcolor{c054} 27.29 \\
0.6 & \cellcolor{c011}  5.83 & \cellcolor{c044} 22.10 & \cellcolor{c049} 24.63 & \cellcolor{c054} 27.17 & \cellcolor{c060} 30.17 \\
0.7 & \cellcolor{c023} 11.52 & \cellcolor{c044} 22.11 & \cellcolor{c053} 26.69 & \cellcolor{c053} 26.84 & \cellcolor{c065} 32.54 \\
0.9 & \cellcolor{c033} 16.87 & \cellcolor{c034} 17.15 & \cellcolor{c042} 21.35 & \cellcolor{c043} 21.92 & \cellcolor{c056} 28.22 \\
1.0 & \cellcolor{c035} 17.92 & \cellcolor{c024} 12.10 & \cellcolor{c028} 14.32 & \cellcolor{c031} 15.60 & \cellcolor{c039} 19.63 \\
1.2 & \cellcolor{c035} 17.96 & \cellcolor{c019}  9.59 & \cellcolor{c021} 10.55 & \cellcolor{c021} 10.91 & \cellcolor{c025} 12.76 \\
1.5 & \cellcolor{c036} 18.02 & \cellcolor{c016}  8.27 & \cellcolor{c017}  8.64 & \cellcolor{c015}  7.98 & \cellcolor{c017}  8.71 \\
1.8 & \cellcolor{c035} 18.05 & \cellcolor{c014}  7.38 & \cellcolor{c015}  7.51 & \cellcolor{c012}  6.27 & \cellcolor{c012}  6.49 \\
2.1 & \cellcolor{c036} 18.09 & \cellcolor{c013}  6.76 & \cellcolor{c013}  6.80 & \cellcolor{c010}  5.20 & \cellcolor{c010}  5.25 \\
2.6 & \cellcolor{c036} 18.21 & \cellcolor{c012}  6.23 & \cellcolor{c012}  6.24 & \cellcolor{c008}  4.49 & \cellcolor{c008}  4.50 \\
3.1 & \cellcolor{c036} 18.33 & \cellcolor{c011}  5.78 & \cellcolor{c011}  5.78 & \cellcolor{c007}  3.94 & \cellcolor{c007}  3.94 \\
3.7 & \cellcolor{c037} 18.59 & \cellcolor{c010}  5.39 & \cellcolor{c010}  5.39 & \cellcolor{c006}  3.47 & \cellcolor{c006}  3.47 \\
4.5 & \cellcolor{c037} 18.82 & \cellcolor{c010}  5.02 & \cellcolor{c010}  5.02 & \cellcolor{c006}  3.02 & \cellcolor{c006}  3.02 \\
5.3 & \cellcolor{c038} 19.08 & \cellcolor{c009}  4.87 & \cellcolor{c009}  4.87 & \cellcolor{c005}  2.61 & \cellcolor{c005}  2.61 \\
6.4 & \cellcolor{c038} 19.32 & \cellcolor{c009}  4.94 & \cellcolor{c009}  4.95 & \cellcolor{c004}  2.26 & \cellcolor{c004}  2.26 \\
7.7 & \cellcolor{c039} 19.58 & \cellcolor{c010}  5.15 & \cellcolor{c010}  5.15 & \cellcolor{c003}  1.95 & \cellcolor{c003}  1.95 \\
9.2 & \cellcolor{c039} 19.82 & \cellcolor{c010}  5.29 & \cellcolor{c010}  5.29 & \cellcolor{c003}  1.69 & \cellcolor{c003}  1.69 \\
\hline
\end{tabular}
\hfill 
\begin{tabular}{|c|r|r|r|r|r|}
\hline
$\delta$ & \multicolumn{1}{c|}{$n$} & \multicolumn{1}{c|}{$|\alpha|$} & \multicolumn{1}{c|}{$|\text{ID}|$} & \multicolumn{1}{c|}{$\gamma ~ \LR$} & \multicolumn{1}{c|}{$\gamma ~ \GR$} \\
\hline
0.5 &  357 & 12.01 &  9 & 60.77 & 84.57 \\
0.6 &  474 & 12.68 &  9 & 62.87 & 83.94 \\
0.7 &  661 & 13.42 & 10 & 60.21 & 80.49 \\
0.9 &  845 & 13.85 & 10 & 46.24 & 70.31 \\
1.0 & 1030 & 14.10 & 11 & 28.25 & 53.14 \\
1.2 & 1242 & 14.31 & 11 & 15.11 & 32.90 \\
1.5 & 1491 & 14.51 & 11 &  7.32 & 16.58 \\
1.8 & 1790 & 14.77 & 11 &  3.03 &  6.60 \\
2.1 & 2149 & 15.07 & 12 &  1.09 &  2.00 \\
2.6 & 2579 & 15.28 & 12 &  0.32 &  0.50 \\
3.1 & 3095 & 15.59 & 12 &  0.08 &  0.08 \\
3.7 & 3715 & 15.92 & 12 &  0.06 &  0.01 \\
4.5 & 4458 & 16.16 & 13 &  0.04 &  0.00 \\
5.3 & 5349 & 16.43 & 13 &  0.06 &  0.00 \\
6.4 & 6419 & 16.76 & 13 &  0.09 &  0.00 \\
7.7 & 7703 & 17.10 & 13 &  0.12 &  0.00 \\
9.2 & 9244 & 17.29 & 14 &  0.15 &  0.00 \\
\hline
\end{tabular}
\hfill
\null
\bigskip
\caption{\label{tab01} Left: Average overhead in percent as a function of density $\delta$; Right: The average number $n$ of nodes, the average address length $|\alpha|$, the size of the IDs ($= \lceil \log_2 n \rceil$), and the average number $\gamma$ of routes that had at least one dead-end in percent as a function of density $\delta$}
\end{table}


The average address lengths and the average number $\gamma$ of routes that have at least one dead-end are shown in Table \ref{tab01}, \ref{tab02}, \ref{tab03}, and \ref{tab04} to the right. These tables show that the sizes of the virtual addresses are only a few bits larger that the sizes of the IDs, which are at least $\lceil \log_2 n \rceil$.

$\BR$ seems to be well suited for resolving the dead-end problem. The advantage of $\BR$ is the very good performance in particular for sparse networks and for networks with obstacles. This show Table \ref{tab02}, \ref{tab03}, and \ref{tab04}.

\begin{table}[hbt]
\center
\small
\hfill
\begin{tabular}{|c|c|c|c|c|c|}
\hline
$\delta$ & $\BR$ & $\LO$ & $\LB$ & $\GO$ & $\GB$ \\
\hline
0.5 & \cellcolor{c003}  1.76 & \cellcolor{c035} 17.64 & \cellcolor{c036} 18.31 & \cellcolor{c050} 25.00 & \cellcolor{c051} 25.53 \\
0.6 & \cellcolor{c005}  2.96 & \cellcolor{c038} 19.49 & \cellcolor{c041} 20.67 & \cellcolor{c052} 26.29 & \cellcolor{c054} 27.45 \\
0.7 & \cellcolor{c012}  6.26 & \cellcolor{c042} 21.32 & \cellcolor{c048} 24.02 & \cellcolor{c054} 27.37 & \cellcolor{c060} 30.40 \\
0.9 & \cellcolor{c024} 12.21 & \cellcolor{c039} 19.99 & \cellcolor{c049} 24.61 & \cellcolor{c053} 26.70 & \cellcolor{c065} 32.67 \\
1.0 & \cellcolor{c033} 16.89 & \cellcolor{c030} 15.31 & \cellcolor{c039} 19.62 & \cellcolor{c044} 22.31 & \cellcolor{c057} 28.98 \\
1.2 & \cellcolor{c036} 18.13 & \cellcolor{c024} 12.49 & \cellcolor{c030} 15.43 & \cellcolor{c034} 17.47 & \cellcolor{c045} 22.64 \\
1.5 & \cellcolor{c036} 18.22 & \cellcolor{c021} 10.80 & \cellcolor{c025} 12.65 & \cellcolor{c027} 13.82 & \cellcolor{c034} 17.35 \\
1.8 & \cellcolor{c036} 18.27 & \cellcolor{c019}  9.94 & \cellcolor{c022} 11.19 & \cellcolor{c022} 11.32 & \cellcolor{c027} 13.73 \\
2.1 & \cellcolor{c036} 18.30 & \cellcolor{c018}  9.35 & \cellcolor{c020} 10.20 & \cellcolor{c019}  9.73 & \cellcolor{c022} 11.49 \\
2.6 & \cellcolor{c036} 18.38 & \cellcolor{c017}  8.72 & \cellcolor{c018}  9.32 & \cellcolor{c016}  8.48 & \cellcolor{c019}  9.73 \\
3.1 & \cellcolor{c036} 18.45 & \cellcolor{c016}  8.19 & \cellcolor{c017}  8.60 & \cellcolor{c014}  7.50 & \cellcolor{c016}  8.41 \\
3.7 & \cellcolor{c037} 18.50 & \cellcolor{c015}  7.87 & \cellcolor{c016}  8.20 & \cellcolor{c013}  6.82 & \cellcolor{c014}  7.49 \\
4.5 & \cellcolor{c037} 18.61 & \cellcolor{c014}  7.48 & \cellcolor{c015}  7.74 & \cellcolor{c012}  6.24 & \cellcolor{c013}  6.74 \\
5.3 & \cellcolor{c037} 18.72 & \cellcolor{c015}  7.54 & \cellcolor{c015}  7.80 & \cellcolor{c011}  5.78 & \cellcolor{c012}  6.19 \\
6.4 & \cellcolor{c037} 18.88 & \cellcolor{c014}  7.36 & \cellcolor{c015}  7.62 & \cellcolor{c010}  5.37 & \cellcolor{c011}  5.72 \\
7.7 & \cellcolor{c037} 18.97 & \cellcolor{c015}  7.81 & \cellcolor{c016}  8.13 & \cellcolor{c010}  5.04 & \cellcolor{c010}  5.34 \\
9.2 & \cellcolor{c038} 19.10 & \cellcolor{c016}  8.24 & \cellcolor{c017}  8.60 & \cellcolor{c009}  4.76 & \cellcolor{c010}  5.02 \\
\hline
\end{tabular}
\hfill
\begin{tabular}{|c|r|r|r|r|r|}
\hline
$\delta$ & \multicolumn{1}{c|}{$n$} & \multicolumn{1}{c|}{$|\alpha|$} & \multicolumn{1}{c|}{$|\text{ID}|$} & \multicolumn{1}{c|}{$\gamma ~ \LR$} & \multicolumn{1}{c|}{$\gamma ~ \GR$} \\
\hline
0.5 &  332 & 11.76 &  9 & 59.26 & 84.50 \\
0.6 &  429 & 12.38 &  9 & 61.74 & 84.49 \\
0.7 &  595 & 13.13 & 10 & 63.19 & 83.25 \\
0.9 &  782 & 13.61 & 10 & 56.57 & 78.09 \\
1.0 &  962 & 14.02 & 10 & 40.57 & 67.01 \\
1.2 & 1161 & 14.26 & 11 & 27.12 & 53.15 \\
1.5 & 1395 & 14.46 & 11 & 17.72 & 39.68 \\
1.8 & 1675 & 14.77 & 11 & 12.08 & 28.91 \\
2.1 & 2011 & 15.05 & 11 &  8.70 & 21.89 \\
2.6 & 2414 & 15.32 & 12 &  6.55 & 16.78 \\
3.1 & 2896 & 15.59 & 12 &  5.14 & 13.17 \\
3.7 & 3475 & 15.99 & 12 &  4.47 & 10.67 \\
4.5 & 4172 & 16.24 & 13 &  3.85 &  8.66 \\
5.3 & 5006 & 16.55 & 13 &  3.80 &  7.35 \\
6.4 & 6008 & 16.93 & 13 &  3.55 &  6.35 \\
7.7 & 7210 & 17.19 & 13 &  3.80 &  5.61 \\
9.2 & 8651 & 17.48 & 14 &  3.97 &  5.00 \\
\hline
\end{tabular}
\hfill
\null
\bigskip
\caption{\label{tab02} Left: Average overhead in percent as a function of density $\delta$ with lakes-mask; Right: Same as Table \ref{tab01} for lakes-mask}
\end{table}


\begin{table}[hbt]
\center
\small
\hfill
\begin{tabular}{|c|c|c|c|c|c|}
\hline
$\delta$ & $\BR$ & $\LO$ & $\LB$ & $\GO$ & $\GB$ \\
\hline
1.2 & \cellcolor{c008}  4.08 & \cellcolor{c043} 21.59 & \cellcolor{c046} 23.23 & \cellcolor{c055} 27.53 & \cellcolor{c058} 29.35 \\
1.5 & \cellcolor{c015}  7.57 & \cellcolor{c046} 23.27 & \cellcolor{c052} 26.50 & \cellcolor{c057} 28.56 & \cellcolor{c064} 32.14 \\
1.8 & \cellcolor{c025} 12.87 & \cellcolor{c045} 22.53 & \cellcolor{c054} 27.46 & \cellcolor{c054} 27.18 & \cellcolor{c066} 33.08 \\
2.1 & \cellcolor{c033} 16.70 & \cellcolor{c036} 18.36 & \cellcolor{c045} 22.89 & \cellcolor{c045} 22.72 & \cellcolor{c057} 28.67 \\
2.6 & \cellcolor{c035} 17.59 & \cellcolor{c029} 14.61 & \cellcolor{c035} 17.67 & \cellcolor{c035} 17.82 & \cellcolor{c044} 22.09 \\
3.1 & \cellcolor{c035} 17.81 & \cellcolor{c025} 12.62 & \cellcolor{c029} 14.58 & \cellcolor{c028} 14.23 & \cellcolor{c033} 16.87 \\
3.7 & \cellcolor{c035} 17.95 & \cellcolor{c023} 11.67 & \cellcolor{c025} 12.96 & \cellcolor{c023} 11.87 & \cellcolor{c026} 13.38 \\
4.5 & \cellcolor{c036} 18.03 & \cellcolor{c022} 11.09 & \cellcolor{c023} 11.98 & \cellcolor{c020} 10.42 & \cellcolor{c022} 11.34 \\
5.3 & \cellcolor{c036} 18.14 & \cellcolor{c021} 10.76 & \cellcolor{c022} 11.42 & \cellcolor{c019}  9.62 & \cellcolor{c020} 10.19 \\
6.4 & \cellcolor{c036} 18.19 & \cellcolor{c020} 10.40 & \cellcolor{c021} 10.87 & \cellcolor{c018}  9.06 & \cellcolor{c018}  9.46 \\
7.7 & \cellcolor{c036} 18.26 & \cellcolor{c020} 10.18 & \cellcolor{c021} 10.53 & \cellcolor{c017}  8.63 & \cellcolor{c017}  8.94 \\
9.2 & \cellcolor{c036} 18.18 & \cellcolor{c019}  9.92 & \cellcolor{c020} 10.20 & \cellcolor{c016}  8.26 & \cellcolor{c017}  8.52 \\
\hline
\end{tabular}
\hfill
\begin{tabular}{|c|r|r|r|r|r|}
\hline
$\delta$ & \multicolumn{1}{c|}{$n$} & \multicolumn{1}{c|}{$|\alpha|$} & \multicolumn{1}{c|}{$|\text{ID}|$} & \multicolumn{1}{c|}{$\gamma ~ \LR$} & \multicolumn{1}{c|}{$\gamma ~ \GR$} \\
\hline
1.2 &  463 & 12.75 &  9 & 62.31 & 83.56 \\
1.5 &  617 & 13.43 & 10 & 63.20 & 82.24 \\
1.8 &  827 & 14.10 & 10 & 59.70 & 77.87 \\
2.1 & 1038 & 14.50 & 11 & 48.16 & 68.16 \\
2.6 & 1260 & 14.83 & 11 & 35.61 & 54.84 \\
3.1 & 1520 & 15.15 & 11 & 25.48 & 40.52 \\
3.7 & 1830 & 15.43 & 11 & 18.41 & 27.57 \\
4.5 & 2199 & 15.79 & 12 & 13.40 & 18.31 \\
5.3 & 2643 & 16.14 & 12 & 10.26 & 12.42 \\
6.4 & 3171 & 16.47 & 12 &  7.67 &  9.08 \\
7.7 & 3807 & 16.85 & 12 &  6.12 &  7.12 \\
9.2 & 4566 & 17.15 & 13 &  4.79 &  5.97 \\
\hline
\end{tabular}
\hfill
\null
\bigskip
\caption{\label{tab03} Left: Average overhead in percent as a function of density $\delta$ with streets-mask; Right: Same as Table \ref{tab01} for streets-mask}
\end{table}


\begin{table}[hbt]
\center
\small
\hfill
\begin{tabular}{|c|c|c|c|c|c|}
\hline
$\delta$ & $\BR$ & $\LO$ & $\LB$ & $\GO$ & $\GB$ \\
\hline
1.8 & \cellcolor{c005}  2.95 & \cellcolor{c048} 24.06 & \cellcolor{c051} 25.55 & \cellcolor{c055} 27.79 & \cellcolor{c058} 29.29 \\
2.1 & \cellcolor{c009}  4.70 & \cellcolor{c050} 25.40 & \cellcolor{c055} 27.76 & \cellcolor{c058} 29.37 & \cellcolor{c064} 32.01 \\
2.6 & \cellcolor{c015}  7.56 & \cellcolor{c052} 26.45 & \cellcolor{c060} 30.27 & \cellcolor{c062} 31.14 & \cellcolor{c071} 35.71 \\
3.1 & \cellcolor{c023} 11.90 & \cellcolor{c054} 27.01 & \cellcolor{c065} 32.89 & \cellcolor{c065} 32.67 & \cellcolor{c080} 40.11 \\
3.7 & \cellcolor{c032} 16.40 & \cellcolor{c054} 27.03 & \cellcolor{c069} 34.64 & \cellcolor{c066} 33.05 & \cellcolor{c086} 43.29 \\
4.5 & \cellcolor{c038} 19.10 & \cellcolor{c052} 26.31 & \cellcolor{c068} 34.17 & \cellcolor{c064} 32.49 & \cellcolor{c087} 43.86 \\
5.3 & \cellcolor{c041} 20.71 & \cellcolor{c051} 25.60 & \cellcolor{c065} 32.98 & \cellcolor{c063} 31.51 & \cellcolor{c086} 43.04 \\
6.4 & \cellcolor{c042} 21.20 & \cellcolor{c049} 24.89 & \cellcolor{c062} 31.38 & \cellcolor{c060} 30.40 & \cellcolor{c082} 41.42 \\
7.7 & \cellcolor{c042} 21.31 & \cellcolor{c049} 24.71 & \cellcolor{c061} 30.87 & \cellcolor{c059} 29.97 & \cellcolor{c080} 40.36 \\
9.2 & \cellcolor{c042} 21.37 & \cellcolor{c049} 24.59 & \cellcolor{c061} 30.57 & \cellcolor{c059} 29.61 & \cellcolor{c079} 39.55 \\
\hline
\end{tabular}
\hfill
\begin{tabular}{|c|r|r|r|r|r|}
\hline
$\delta$ & \multicolumn{1}{c|}{$n$} & \multicolumn{1}{c|}{$|\alpha|$} & \multicolumn{1}{c|}{$|\text{ID}|$} & \multicolumn{1}{c|}{$\gamma ~ \LR$} & \multicolumn{1}{c|}{$\gamma ~ \GR$} \\
\hline
1.8 &  816 & 14.17 & 10 & 68.37 & 88.43 \\
2.1 &  992 & 14.65 & 10 & 68.34 & 87.93 \\
2.6 & 1236 & 15.08 & 11 & 67.71 & 87.58 \\
3.1 & 1565 & 15.57 & 11 & 66.05 & 87.19 \\
3.7 & 1971 & 16.08 & 11 & 63.91 & 86.45 \\
4.5 & 2431 & 16.49 & 12 & 61.17 & 85.46 \\
5.3 & 2976 & 16.82 & 12 & 58.51 & 84.32 \\
6.4 & 3621 & 17.21 & 12 & 56.30 & 83.07 \\
7.7 & 4404 & 17.51 & 13 & 55.12 & 82.15 \\
9.2 & 5329 & 17.91 & 13 & 54.50 & 81.36 \\
\hline
\end{tabular}
\hfill
\null
\bigskip
\caption{\label{tab04} Left: Average overhead in percent as a function of density $\delta$ with  buildings-mask and only visible connections; Right: Same as Table \ref{tab01} for buildings-mask and only visible connections}
\end{table}


\subsection{Coarsening the edge weights}

During the initialization phase, the nodes have to store and manage distance values to landmark nodes. The distances are built by sums of edge weights. Our weights of the connections defined in equation \ref{equ01} can simply be coarsed to only $k$ different values by
\begin{equation}
\label{equ02}
\omega_k(u,v) = \left\lceil k \cdot \frac{400 + d^e(u,v)^2}{400 + r^2} \right\rceil.
\end{equation}
If we use weight function $\omega_k$ instead of $\omega$, the overhead of $\BR$ will increase, depending on the number $k$ of different weights for the connections. In Table \ref{tab09} it is shown how the overhead increases when the number of different weights decreases. If we have the same weight for all connections, the algorithm considers hop distances. Table \ref{tab09} to the left shows that 8 different weights for the connections will decrease the overhead of $\BR$ only by a few percent. The table also shows that the denser the network is, the larger the overhead becomes. However, if the network is dense, the greedy algorithms reach a dead-end less often. If $\BR$ is used only to guarantee delivery for greedy routing protocols, the overhead will increase only minimal.

\begin{table}[hbt]
\center
\small
\hfill
\begin{tabular}{|c|r|r|r|r|r|r|}
\hline
$\delta$ & \multicolumn{1}{c|}{$\omega$} & \multicolumn{1}{c|}{$\omega_{16}$} & \multicolumn{1}{c|}{$\omega_{8}$} & \multicolumn{1}{c|}{$\omega_{4}$} & \multicolumn{1}{c|}{$\omega_{2}$} & \multicolumn{1}{c|}{$\omega_1$} \\
\hline
 0.5 &  2.94 &  3.02 &  3.25 &  4.05 &  5.08 &  7.44 \\
 0.6 &  5.83 &  5.89 &  6.22 &  7.32 &  8.67 & 11.40 \\ 
 0.7 & 11.52 & 11.77 & 12.22 & 13.68 & 15.26 & 19.82 \\
 0.9 & 16.87 & 17.06 & 17.62 & 19.02 & 21.79 & 27.62 \\
 1.0 & 17.92 & 18.43 & 19.12 & 21.37 & 23.95 & 31.33 \\
 1.2 & 17.96 & 18.47 & 19.43 & 22.37 & 25.16 & 34.49 \\
 1.5 & 18.02 & 18.62 & 19.83 & 23.21 & 26.06 & 37.34 \\
 1.8 & 18.05 & 18.88 & 20.27 & 24.12 & 26.85 & 40.39 \\
 2.1 & 18.09 & 19.17 & 20.76 & 25.01 & 27.37 & 43.21 \\
 2.6 & 18.21 & 19.49 & 21.26 & 25.93 & 27.71 & 45.78 \\
 3.1 & 18.33 & 19.86 & 21.78 & 26.66 & 27.87 & 48.22 \\
 3.7 & 18.59 & 20.26 & 22.22 & 27.29 & 27.90 & 50.41 \\
 4.5 & 18.82 & 20.67 & 22.54 & 27.76 & 27.92 & 52.30 \\
 5.3 & 19.08 & 21.06 & 23.04 & 28.09 & 28.02 & 53.97 \\
 6.4 & 19.32 & 21.43 & 23.30 & 28.31 & 28.09 & 55.40 \\
 7.7 & 19.58 & 21.75 & 23.53 & 28.45 & 28.26 & 56.80 \\
 9.2 & 19.82 & 22.02 & 23.67 & 28.46 & 28.54 & 57.67 \\
\hline
\end{tabular}
\hfill
\begin{tabular}{|c|r|r|}
\hline
 $|\alpha|$ & transmissions & sub-networks \\
\hline
    &  82205 &    1 \\
 \hline
  1 & 137658 &    1 \\ 
  2 & 164885 &    2 \\
  3 & 149594 &    4 \\
  4 & 123651 &    8 \\
  5 & 115242 &   16 \\
  6 & 109842 &   32 \\
  7 & 108777 &   64 \\
  8 & 101131 &  128 \\
  9 &  99712 &  255 \\
 10 &  95281 &  466 \\
 11 &  84330 &  699 \\
 12 &  55758 &  561 \\
 13 &  21186 &  234 \\
 14 &   4668 &   54 \\
 15 &    366 &    5 \\
\hline
\end{tabular}
\hfill
\null
\bigskip
\caption{\label{tab09} Left: Average overhead for coarsened edge weights, Right: Number of transmissions during the preprocessing phase for a network with 2500 nodes and 23462 edges}
\end{table}

\subsection {Energy to setup the data structure}

During the setup of the initial data structure, the network is flooded several times. To determine the distances to the arbitrarily chosen start node $w$, the complete network is flooded. This is also the case for determining the distances to the landmark nodes $x_0$ and $x_1$. After that the network is flooded only partially, because the distances to landmark nodes $x_{00}$ and $x_{01}$ ($x_{10}$ and $x_{11}$) are only relevant for the nodes having a virtual address starting with $0$ (with $1$, respectively). In general, the total amount of energy to determine all distances to the $2 \cdot 2^k$ landmark nodes with the same address length $k$ is less than the amount of energy to flood the network two times. Our experimental analyses show that the amount of energy necessary for the next bipartition decreases for increasing address lengths.


Table \ref{tab09} to the right shows the number of transmissions to build the initial data structure. This is not the average over 1000 networks but only one typical example. The network is created with a node density of $2.5 \cdot 10^{-3}$ nodes per $m^2$ without any masks. It has 2500 nodes and 23462 edges. The first line shows the total number of transmissions for flooding the network from the arbitrarily chosen node $w$. The second line shows the total number of transmissions for flooding the network from the two landmark nodes $x_0$ and $x_1$. The third line shows the total number of transmissions for flooding two sub-networks from four landmark nodes $x_{00}$, $x_{01}$, $x_{10}$, and $x_{11}$, and so on.

\section {Conclusions}

Greedy algorithms for routing in wireless ad hoc sensor networks are easy to implement. They are very effective but can unfortunately reach an impasse. In this paper, we have introduced and analyzed a very simple hierarchical bipartition technique for wireless ad hoc sensor networks. Every node gets a unique virtual address that can be used to route through the network with delivery guarantee. To keep the advantages of greedy routing, we suggest to use $\BR$ for finding the way out of a dead-end. This is especially very interesting if the network is sparse or contains obstacles. In these cases, the probability to reach a dead-end is very high. Our experimental evaluations even show the following: If more than 50\% of the routes reach at least one dead-end, then the performance of stand-alone $\BR$ is in general better than the performance of geographic greedy routing with a shortest path dead-end handling.

The main assets and drawbacks of HBR are the following:

{\center
\begin{tabular}{p{220pt}p{220pt}}
Assets: & Drawbacks: \\
\begin{enumerate}
\item packet delivery guaranty
\item no packet overhead, because the unique virtual addresses are of size $\log n$ on average
\item small routing tables
\item no geographic coordinates necessary
\end{enumerate} &
\begin{enumerate}
\item works only for static network structures and is vulnerable to even small changes of topology
\item relative time-consuming and energy-consuming set-up phase
\item unbalanced partitioning process results in large address lengths
\item unbounded worst-case stretch factor
\end{enumerate}
\end{tabular}}

\bibliographystyle{plain}
\bibliography{../Sensornetzwerke.bib}

\begin{thebibliography}{10}

\bibitem{DBLP:conf/infocom/2005}
{\em INFOCOM 2005. 24th Annual Joint Conference of the IEEE Computer and
  Communications Societies, 13-17 March 2005, Miami, FL, USA}. IEEE, 2005.

\bibitem{BPDCFS06}
F~Benbadis, J.-J. Puig, M.~Dias de~Amorim, C.~Chaudet, T.~Friedman, and
  D.~Simplot-Ryl.
\newblock Jumps: Enhancing hop-count positioning in sensor networks using
  multiple coordinates.
\newblock {\em CoRR}, abs/cs/0604105, 2006.

\bibitem{BCGKOSU04}
P.~Boone, E.~Ch{\'a}vez, L.~Gleitzky, E.~Kranakis, J.~Opatrny, G.~Salazar, and
  J.~Urrutia.
\newblock Morelia test: Improving the efficiency of the gabriel test and face
  routing in ad-hoc networks.
\newblock In R.~Kralovic and O.~S{\'y}kora, editors, {\em SIROCCO}, volume 3104
  of {\em Lecture Notes in Computer Science}, pages 23--34. Springer, 2004.

\bibitem{BMSU01}
P.~Bose, P.~Morin, I.~Stojmenovic, and J.~Urrutia.
\newblock Routing with guaranteed delivery in ad hoc wireless networks.
\newblock {\em Wireless Networks}, 7(6):609--616, 2001.

\bibitem{CCDU05}
A.~Caruso, S.~Chessa, S.~De, and A.~Urpi.
\newblock Gps free coordinate assignment and routing in wireless sensor
  networks.
\newblock In {\em INFOCOM\/} \cite{DBLP:conf/infocom/2005}, pages 150--160.

\bibitem{CMT07}
E.~Ch{\'a}vez, N.~Mitton, and H.~Tejeda.
\newblock Routing in wireless networks with position trees.
\newblock In E.~Kranakis and J.~Opatrny, editors, {\em ADHOC-NOW}, volume 4686
  of {\em Lecture Notes in Computer Science}, pages 32--45. Springer, 2007.

\bibitem{CWLG97}
C.~Chiang, H.~Wu, W.~Liu, and M.~Gerla.
\newblock Routing in clustered multihop, mobile wireless networks.
\newblock In {\em IEEE Singapore International Conference on Networks}, pages
  197--211, 1997.

\bibitem{CJLMQV03}
T.~Clausen, P.~Jacquet, A.~Laouiti, P.~Muhlethaler, A.~Qayyum, and L.~Viennot.
\newblock Optimized link state routing protocol ({OLSR}).
\newblock {\em RFC 3626}, 2003.

\bibitem{EMS07}
E.~H. Elhafsi, N.~Mitton, and D.~Simplot-Ryl.
\newblock Cost over progress based energy efficient routing over virtual
  coordinates in wireless sensor networks.
\newblock In {\em WOWMOM}, pages 1--6. IEEE, 2007.

\bibitem{EMS08}
E.H. Elhafsi, N.~Mitton, and D.~Simplot-Ryl.
\newblock End-to-end energy efficient geographic path discovery with guaranteed
  delivery in ad hoc andsensor networks.
\newblock In {\em PIMRC}, pages 1--5. IEEE, 2008.

\bibitem{FGGSZ05}
Q.~Fang, J.~Gao, L.J. Guibas, V.~de~Silva, and L.~Zhang.
\newblock {GLIDER}: gradient landmark-based distributed routing for sensor
  networks.
\newblock In {\em INFOCOM\/} \cite{DBLP:conf/infocom/2005}, pages 339--350.

\bibitem{Fin87}
G.G. Finn.
\newblock Routing and addressing problems in large metropolitan scale
  internetworks.
\newblock Technical Report ISU/RR-87-180, Internetworks,ISI, 1987.

\bibitem{FRZ05}
R.~Fonseca, S.~Ratnasamy, J.~Zhao, C.T. Ee, D.E. Culler, S.~Shenker, and
  I.~Stoica.
\newblock Beacon vector routing: Scalable point-to-point routing in wireless
  sensornets.
\newblock In {\em NSDI}. USENIX, 2005.

\bibitem{FS06}
H.~Frey and I.~Stojmenovic.
\newblock On delivery guarantees of face and combined greedy-face routing in ad
  hoc and sensor networks.
\newblock In M.~Gerla, C.~Petrioli, and R.~Ramjee, editors, {\em MOBICOM},
  pages 390--401. ACM, 2006.

\bibitem{HL86}
T.C. Hou and V.~Li.
\newblock Transmission range control in multihop packet radio networks.
\newblock In {\em IEEE Transactions on Communications}, volume~34, pages
  38--44. IEEE, IEEE, January 1986.

\bibitem{IS09}
K.~Iwanicki and M.~van Steen.
\newblock On hierarchical routing in wireless sensor networks.
\newblock In {\em IPSN}, pages 133--144. ACM, 2009.

\bibitem{JMB01}
D.B. Johnson, D.A. Maltz, and J.~Broch.
\newblock {DSR}: The dynamic source routing protocol for multi-hop wireless ad
  hoc networks.
\newblock In C.E. Perkins, editor, {\em Ad Hoc Networking}, pages 139--172.
  Addison-Wesley, 2001.

\bibitem{KSU99}
E.~Kranakis, H.~Singh, and J.~Urrutia.
\newblock Compass routing on geometric networks.
\newblock In {\em CCCG}, pages 51--54, 1999.

\bibitem{KCVP95}
P.~Krishna, M.~Chatterjee, N.H. Vaidya, and D.K. Pradhan.
\newblock A cluster-based approach for routing in ad-hoc networks.
\newblock In {\em Symposium on Mobile and Location-Independent Computing},
  pages 1--10. USENIX, 1995.

\bibitem{KAE00}
S.~Kumar, C.~Alaettinoglu, and D.~Estrin.
\newblock Scalable object-tracking through unattended techniques (scout).
\newblock In {\em ICNP}, pages 253--262, 2000.

\bibitem{KNS06}
J.~Kuruvila, A.~Nayak, and I.~Stojmenovic.
\newblock Progress and location based localized power aware routing for ad hoc
  and sensor wireless networks.
\newblock {\em IJDSN}, 2(2):147--159, 2006.

\bibitem{LGSM04}
J.~Li, L.~Gewali, H.~Selvaraj, and M.~Venkatesan.
\newblock Hybrid greedy/face routing for ad-hoc sensor network.
\newblock In {\em DSD}, pages 574--578. IEEE, 2004.

\bibitem{LA08}
K.~Liu and N.B. Abu-Ghazaleh.
\newblock Stateless and guaranteed geometric routing on virtual coordinate
  systems.
\newblock In {\em MASS}, pages 340--346. IEEE, 2008.

\bibitem{MRSS08}
N.~Mitton, T.~Razafindralambo, D.~Simplot-Ryl, and I.~Stojmenovic.
\newblock Hector is an energy efficient tree-based optimized routing protocol
  for wireless networks.
\newblock In {\em MSN}, pages 31--38. IEEE, 2008.

\bibitem{PBD03}
C.~Perkins, E.~Belding-Royer, and S.~Das.
\newblock Ad hoc on-demand distance vector ({AODV}) routing.
\newblock {\em RFC 3561}, 2003.

\bibitem{SR06}
J.A. S{\'a}nchez and P.M. Ruiz.
\newblock Exploiting local knowledge to enhance energy-efficient geographic
  routing.
\newblock In J.~Cao, I.~Stojmenovic, X.~Jia, and S.K. Das, editors, {\em MSN},
  volume 4325 of {\em Lecture Notes in Computer Science}, pages 567--578.
  Springer, 2006.

\bibitem{Sto06}
I.~Stojmenovic.
\newblock Localized network layer protocols in wireless sensor networks based
  on optimizing cost over progress ratio.
\newblock {\em IEEE Network}, 20(1):21--27, 2006.

\bibitem{SD04}
I.~Stojmenovic and S.~Datta.
\newblock Power and cost aware localized routing with guaranteed delivery in
  unit graph based ad hoc networks.
\newblock {\em Wireless Communications and Mobile Computing}, 4(2):175--188,
  2004.

\bibitem{SL01}
I.~Stojmenovic and X.~Lin.
\newblock Power-aware localized routing in wireless networks.
\newblock {\em IEEE Transactions Parallel Distrib. Syst.}, 12(11):1122--1133,
  2001.

\bibitem{TK84}
H.~Takagi and L.~Kleinrock.
\newblock Optimal transmission ranges for randomly distributed packet radio
  terminals.
\newblock {\em IEEE Transactions on Communications}, 32(3):246--257, 1984.

\bibitem{Tou80}
G.T. Toussaint.
\newblock The relative neighbourhood graph of a finite planar set.
\newblock {\em Pattern Recognition}, 12(4):261--268, 1980.

\bibitem{Tsu88}
P.F. Tsuchiya.
\newblock The landmark hierarchy: a new hierarchy for routing in very large
  networks.
\newblock In {\em SIGCOMM}, pages 35--42, 1988.

\bibitem{WSWLD06}
Y.~Wangy, W.-Z. Songz, W.~Wang, X.-Y. Li, and T.A. Dahlbergy.
\newblock {LEARN}: Localized energy aware restricted neighborhood routing for
  ad hoc networks.
\newblock In {\em SECON}, pages 508--517. IEEE, 2006.

\bibitem{W02}
J.~Wu.
\newblock Dominating-set-based routing in ad hoc wireless networks.
\newblock In I.~Stojmenovic, editor, {\em Handbook of Wireless Networks and
  Mobile Computing, chapter 20}, pages 425--450. John Wiley \& Sons, 2002.

\bibitem{ZS07}
H.~Zhang and H.~Shen.
\newblock Eegr: Energy-efficient geographic routing in wireless sensor
  networks.
\newblock In {\em ICPP}, pages 67--75. IEEE Computer Society, 2007.

\end{thebibliography}

\end{document}